\documentclass{article}
\def\nkyear{2024}
\def\nkvol{xx}
\def\nkno{xx}

\def\firstpage{1}
\def\lastpage{23}
\def\correction{4}
\def\shorttitle{Gain-only neural PDE Backstepping control} % —πÀı±ÍÃ‚, æ°¡ø≤ª≥¨π˝60∏ˆ◊÷∑˚ (∞¸¿®ø’∏Ò)
\def\shortauthor{{\it R. Vazquez\ \ and\ \ M. Krstic}} % ◊˜’ﬂ
\usepackage b \baselineskip 14.5pt
 \usepackage{amssymb}
 \usepackage{amsmath}
  \usepackage{url}
 \usepackage{cite}
 \usepackage{amsfonts}
\newtheorem{myth}{Theorem}
\input{cam.doi}
\journalnumber{11401}
\DOImsnr{0001}
 \DOIyear{2007}

\title{\Large \bf \boldmath\ \\ GAIN-ONLY NEURAL OPERATORS\\ FOR PDE BACKSTEPPING$^{\ast}$}
\author{\large  Rafael VAZQUEZ$^1$\qquad Miroslav KRSTIC$^{2}$} % ◊˜’ﬂ»´√˚

\date{}

\begin{document}

\maketitle

\thispagestyle{first}
\renewcommand{\thefootnote}{\fnsymbol{footnote}}

% µ•Œª°¢µÿ÷∑°¢ª˘Ω
\footnotetext{\hspace*{-5mm} \begin{tabular}{@{}r@{}p{13.4cm}@{}}
& Manuscript received  1 Dec 2024\\ 
$^1$ & Department of Aerospace Engineering, Universidad de Sevilla, Camino de los Descubrimiento s.n., 41092 Sevilla, Spain.\\
&{E-mail:} rvazquez1@us.es \\
$^{2}$ & Department of Mechanical and Aerospace Engineering, University of California San Diego, La Jolla, CA 92093-0411, USA.\\
&{E-mail:} krstic@ucsd.edu \\
 $^{\ast}$ & R. Vazquez was supported by grant TED2021-132099B-C33 funded by MCIN/ AEI/ 10.13039 /501100011033 and by ``European Union NextGenerationEU/PRTR.''.
\end{tabular}}

\renewcommand{\thefootnote}{\arabic{footnote}}

\begin{abstract} 
In this work we advance the recently-introduced deep learning-powered approach to PDE backstepping control by proposing a method that approximates only the control gain function---a function of one variable---instead of the entire kernel function of the backstepping transformation, which depends on two variables. This idea is introduced using several benchmark unstable PDEs, including hyperbolic and parabolic types, and extended to 2×2 hyperbolic systems. By employing a backstepping transformation that utilizes the exact kernel (suitable for gain scheduling) rather than an approximated one (suitable for adaptive control), we alter the quantification of the approximation error. This leads to a significant simplification in the target system, shifting the perturbation due to approximation from the domain to the boundary condition. Despite the notable differences in the Lyapunov analysis, we are able to retain stability guarantees with this simplified approximation approach. Approximating only the control gain function simplifies the operator being approximated and the training of its neural approximation, potentially reducing the neural network size. The trade-off for these simplifications is a more intricate Lyapunov analysis, involving higher Sobolev spaces for some PDEs, and certain restrictions on initial conditions arising from these spaces. It is crucial to carefully consider the specific requirements and constraints of each problem to determine the most suitable approach; indeed, recent works have demonstrated successful applications of both full-kernel and gain-only approaches in adaptive control and gain scheduling contexts.
\vskip 4.5mm

\nd \begin{tabular}{@{}l@{ }p{10.1cm}} {\bf Keywords } & PDE control, backstepping method, neural operators, gain approximation, hyperbolic systems, parabolic systems, Lyapunov stability
\end{tabular}

\nd {\bf 2000 MR Subject Classification } % ∑÷¿‡∫≈
35Q93, 93C20, 93D15

\end{abstract}

\baselineskip 14pt

\setlength{\parindent}{1.5em}

\setcounter{section}{0}

\section{Introduction}

\subsection{Deep learning-powered PDE backstepping}
Recent advancements in the control of partial differential equations (PDEs) have leveraged deep learning techniques to accelerate the computation of gains for model-based control laws via the backstepping method~\cite{bhan2023neural, krstic2023neural, qi2023neural, wang2023deep}. These developments are based on a novel neural network architecture and its mathematical foundation known as DeepONet~\cite{Lu2021}. This approach enables the generation of accurate approximations of nonlinear operators, capturing solutions derived from PDEs that define controller gains. The DeepONet framework extends the "universal approximation theorem" for functions~\cite{Cyb1989, Hor1989} to provide a universal approximation for nonlinear operators~\cite{Che1995, Li2020a, Li2020b, Lu2019, Lu2021}.

For certain PDEs related to the computation of stabilizing control laws---such as the kernel equations in the backstepping method---any change in the plant parameter functions requires solving a new problem. Using the deep learning approach, it suffices to recompute the backstepping kernels through a pre-trained DeepONet map, which are then evaluated at specific values to obtain the gain kernels. Therefore, this neural architecture encapsulates the transformation from the plant's functional coefficients to the controller gain kernels, simplifying gain computation to mere function evaluation and eliminating the need for numerically solving the kernel equations. This methodology for PDE control maintains the theoretical guarantees of a nominal closed-loop system in its approximate form and has been developed for both delay-free hyperbolic~\cite{bhan2023neural} and parabolic~\cite{krstic2023neural} PDEs, as well as for delay-compensating control of hyperbolic partial integro-differential equations (PIDEs)\cite{qi2023neural} and reaction-diffusion PDEs\cite{wang2023deep}.

\subsection{Contribution of this work}
This work introduces a novel methodology to directly approximate backstepping gains using neural operators, bypassing the need to approximate the full kernels as in prior studies. This approach is demonstrated through several simple cases: two parabolic examples (with Dirichlet and Neumann boundary conditions, respectively), a one-dimensional hyperbolic PIDE, and finally a $2\times2$ 1-D hyperbolic system of equations.

The core idea is to employ the exact kernel in the backstepping transformation, instead of an approximated one, and thereby represent the discrepancy from the approximation as a perturbation at the boundary of the target system (resulting from the unknown exact backstepping transformation), rather than as a perturbation within the domain. This alternative approach introduces a more ``unforgiving'' perturbation and adds complexity to the stability analysis. However, by accepting this analytical challenge, we potentially reduce the approximation burden in terms of training set computation and neural network size, due to the elimination of previous constraints on the approximate kernel.

The magnitude of the perturbation at the target system's boundary is directly influenced by the quality of the approximation, which is governed by the universal approximation theorem. The necessary conditions for the approximation are established through a robust exponential stability analysis using functional norms appropriate for a nonlocal perturbation acting at a boundary condition. In the examples presented in the paper, using the $L^2$ norm suffices, except for the reaction-diffusion equation with Dirichlet boundary conditions, which requires the $H^1$ norm.

The results achieved in this work permit arbitrary decay rates. Notably, in the parabolic case, increasing the decay rate necessitates both a different gain and a higher approximation quality. In contrast, for the hyperbolic case, the same gain leads to an increased decay rate if it is better approximated, owing to the finite-time decay of the target system with an exact kernel.

Given that the backstepping approach has become a ubiquitous technique in PDE control, advancements in gain approximation using learning methods, such as the one we propose, could potentially be extended to many families of systems. Originally introduced for feedback control in one-dimensional reaction-diffusion PDEs~\cite{krstic}, backstepping has expanded to multi-dimensional applications~\cite{nball} and has been utilized in various systems like flow control~\cite{vazquez,vazquez-coron}, heat loops~\cite{convloop}, and others~\cite{Rijke, vazquez2, vazquez-nonlinear, florent, krstic2008backstepping, jie, krstic2, guangwei, krstic5}. Since the approach presented in this paper is based on reusing both the target system and the Lyapunov function from the original design, the potential for extending the method is significant, albeit with some caveats, as explained next.

This paper serves as an extended and improved version of our previous conference paper~\cite{vazquez2024gain}, providing a more comprehensive treatment of the gain-only approach.

\subsection{Comparison of approaches}
The methodology presented in this paper differs from the full-kernel approximation strategies outlined in~\cite{krstic2023neural,bhan2023neural,qi2023neural,wang2023deep} in several key aspects, with each approach offering unique advantages and challenges. Hereafter, we will refer to the approach in this paper as the "gain-only approach," and the latter as the "full-kernel approach."

A salient advantage of the gain-only approach lies in its focus on approximating an inherently 1-dimensional gain kernel. This specificity paves the way for a more tractable loss function during training, streamlining the neural network design by reducing the number of hyperparameters. Such an approach not only necessitates a diminished training set size, attributed to the output function's singular argument, but also should lead to a marked reduction in the required training duration. Additionally, the analytical calculations needed to be performed to derive the "perturbed" target system are more straightforward, as backstepping transformation that employs the exact kernel circumvents the intricacies of derivatives or traces of the approximated kernels. The gain-only methodology also accommodates a broader range of coefficient smoothness, facilitating the inclusion of non-differentiable functions, even though functions should be at least Lipschitz continous for the universal approximation theorem to hold. In comparison, the full-kernel method entails more demanding requirements for the plant coefficients.

However, the simplicity of the gain-only approach comes with certain drawbacks. The Dirichlet parabolic cases under this approach require an $H^1$ analysis, which is inherently more challenging than an $L^2$ analysis, and necessitates smoother initial conditions and compatibility conditions. Additionally, in all cases, the need for bounds on inverse kernels may present a challenge. Although these bounds can be inferred indirectly from the direct kernel's bounds, the results are often conservatively estimated.

In situations where kernel approximation is critical to design considerations, such as in adaptive control~\cite{krstic4}, one might anticipate that only the full-kernel approach is applicable. However, recent studies have successfully employed both full-kernel and gain-only approaches in this context. In~\cite{lamarque2024adaptive}, the authors first demonstrate the applicability of the full-kernel approach in adaptive control, where designing the update law requires a target system based on the known approximate kernel. They also present an alternative method using a passive identifier, which enables the gain-only approach to be used in adaptive control without needing to approximate the kernel's derivative.

Gain scheduling (GS) is another control technique that can benefit from kernel approximation. GS adjusts the controller gains based on the system's current operating conditions, enhancing performance across a wide range of operating points. In the context of PDEs, GS involves treating the plant coefficients as quasi-constant and updating the controller gains accordingly. The gain-only approach has been successfully applied to PDE GS in~\cite{lamarque2024gain}. In that work, the target system under GS involves nonlinear perturbations resulting from both the quasi-constant treatment of plant coefficients and the approximation of the kernel. The gain-only approach simplifies the GS design by eliminating the perturbations due to kernel approximation while still managing the perturbations arising from the quasi-constant treatment of plant coefficients, demonstrating its effectiveness in GS applications.

While the gain-only approach offers several significant improvements, it is essential to fully comprehend the trade-offs in specific contexts, as demonstrated by these recent works~\cite{lamarque2024adaptive,lamarque2024gain}. The choice between full-kernel and gain-only approaches depends on the specific requirements and constraints of the problem at hand, and researchers should carefully consider the implications of each method when designing their control strategies.

\subsection{Generating the training set}
It should be noted that both methods face the same challenge in generating the training set, since the controller gains (functions of a single argument, or "1D functions") are always obtained as traces of the full backstepping kernels. Therefore, the numerical computation of the full 2D backstepping kernel cannot be bypassed when training the 1D control gain functions.

The backstepping kernel equations are typically linear hyperbolic PDEs defined on a specific triangular domain, described by Goursat~\cite{holten}, with unique boundary conditions. The numerical solution of these equations is essential for generating the training set, but it can be challenging as the topic of numerical solution of Goursat PDEs for PDE backstepping has not been extensively addressed in the literature. Hints on numerical algorithms are scattered across various sources~\cite{krstic, simon, deutscher, aamo, convloop, jad, auriol2}. Advanced methods for Goursat problems ~\cite{day}, however, have not been utilized for backstepping kernel equations, and adapting these techniques can be complex, especially when dealing with discontinuities.  A new, more general method based on power series approximations has been recently developed~\cite{vazquez-cdc2023}, and its extension to MATLAB~\cite{lin2024towards} appears promising as a tool for generating training sets.

\subsection{Structure of this paper}
The gain approximation approach proposed in this paper is introduced through several examples, each concluding with a theorem that provides the conditions under which the feedback law, using the gain approximated by a neural operator, ensures exponential stability. We begin in Section~\ref{sect-PIDE} with the simplest possible example, a one-dimensional hyperbolic PIDE plant. Then, in Section~\ref{sect-RD}, we consider two reaction-diffusion cases---the Dirichlet case and the Neumann case---which are treated in parallel. Finally, in Section~\ref{sect-hypsist}, we analyze a hyperbolic $2 \times 2$ system. We conclude in Section~\ref{sect-concl} with some final remarks.

\section{1-D hyperbolic PIDE} \label{sect-PIDE}
Consider the plant
\begin{eqnarray}\label{hyp1}
u_t&=&u_x+g(x) u(0,t)+\int_0^x f(x,y) u(y)dy, \\
u(1,t)&=&U(t).\label{hyp2}
\end{eqnarray}
where $u(x,t)$ is the state, $U(t)$ the actuation, $x \in[0,1]$, $t>0$,
for $f\in{\cal C}^0\left({\cal T}\right)$ and $g\in{\cal C}^0\left([0,1]\right)$, where ${\cal T}=\left\{(x,y):0\leq y \leq x \leq 1\right\}$.

\subsection{Backstepping feedback law design for 1-D hyperbolic PDEs}
The backstepping method is based on the  use of a direct/inverse backstepping transformation pair
\begin{eqnarray}
w(x,t)&=&u(x,t)-\int_0^x K(x,\xi)u(\xi,t) d\xi,\label{eqn-hyptrans}\\
u(x,t)&=&w(x,t)+\int_0^x L(x,\xi)w(\xi,t) d\xi,\label{eqn-hypinvtrans}
\end{eqnarray}
where $K(x,\xi)$ and $L(x,\xi)$, are, respectively, the direct and inverse backstepping kernels, that verify hyperbolic PDEs (the kernel equations)~\cite{krstic2008backstepping} involving the coefficients  $f$ and $g$ of the system, namely
\begin{eqnarray}
K_x+K_\xi&=&\int_\xi^x K(x,s) f(s,\xi) ds -f(x,\xi),\label{eqn-hypk1}\\
K(x,0)&=&\int_0^x K(x,s) g(s) ds-g(x),\label{eqn-hypk2}
\end{eqnarray}
in the domain ${\cal T}$,
and a very similar equation for $L(x,\xi)$.

Applying (\ref{eqn-hyptrans}) to (\ref{hyp1})--(\ref{hyp2}), with the kernel satisfying (\ref{eqn-hypk1})--(\ref{eqn-hypk2}), one can find that the PDE verified by the new $w(x,t)$ variable, the target system, is
\begin{eqnarray}
w_t&=&w_x,\\
w(1,t)&=&U(t)-\int_0^1 K(1,\xi)u(\xi,t) d\xi.\label{eqn-targethypbc}
\end{eqnarray}
Thus, defining the feedback gain in terms of the backstepping kernel as
\begin{equation}
K_1(\xi)=K(1,\xi)\,,    
\end{equation} 
and using the control law 
\begin{equation}
    U(t)=\int_0^1 K_1(\xi)u(\xi,t) d\xi\,, 
\end{equation}
one achieves a homogeneous boundary condition $w(1,t)=0$ in (\ref{eqn-targethypbc}) and, consequently,  the target system becomes exponentially stable (in fact convergent to zero in finite time 1 due to the unity transport speed). Since the backstepping transformation is invertible, this implies exponential stability in the original plant coordinates $u(x,t)$, see~\cite{krstic2008backstepping} for details, where the following result, later important to our work, is stated and proven.
\begin{myth} \label{th-kernelboundhyp}
Consider the equations verified by $K(x,\xi)$ (given by (\ref{eqn-hypk1})--(\ref{eqn-hypk2})) and $L(x,\xi)$ (given in~\cite{krstic2008backstepping}) in the domain ${\cal T}$ with $f\in{\cal C}^0\left({\cal T}\right)$ and $g\in{\cal C}^0\left([0,1]\right)$.  Then, there exists a unique solution $K,L \in {\cal C}^1\left({\cal T}\right)$, and denoting $\bar g=\Vert g \Vert_{\infty}=\max_{x\in[0,1]} \vert g(x)\vert$, $\bar f=\Vert f \Vert_{\infty}=\max_{(x,y)\in{\cal T}} \vert f(x,y)\vert$ and $\Vert K \Vert_{\infty}=\max_{(x,\xi)\in{\cal T}} \vert K(x,\xi)\vert$ (and similarly for $L$), one has
\begin{equation}
\Vert L \Vert_{\infty},\Vert K \Vert_{\infty} \leq \left(\bar f+ \bar g\right)\mathrm{e}^{\bar f+ \bar g}    
\end{equation}
\end{myth}
\subsection{Accuracy of approximation of backstepping 1-D hyperbolic gain operators with DeepONet}\label{sect-accuracy}
The main idea of this work (compared with the previous results~\cite{krstic2023neural,bhan2023neural,qi2023neural,wang2023deep}) is to directly approximate the gain operator $K_1(\xi)$ by DeepONet as $\hat K_1(\xi)$ and thus apply an approximate feedback law $U(t)=\int_0^1 \hat K_1(\xi)u(\xi,t) d\xi$.

Let the operator ${\cal K}_1:{\cal C}^0\left({\cal T}\right)\times {\cal C}^0\left([0,1]\right)\rightarrow {\cal C}^1\left([0,1]\right)$ be given by
\begin{equation}
K_1(x)=: {\cal K}_1(f,g)(x)
\end{equation}
By applying the DeepONet universal approximation Theorem (see~\cite[Theorem 2.1]{deng2022approximation}), we get the following key result for the
approximation of the backstepping kernel gain (the  ${\cal K}_1$ operator) by a DeepONet (see~\cite{bhan2023neural} for the exact definition of a neural operator). The proof of continuity and Lipschitzness is obtained by mimicking the successive approximation calculation in the proof of Theorem~\ref{th-kernelboundhyp}. 
\begin{myth}\label{th-approx}
For all $B_f$, $B_g >0$ and $\epsilon>0$, there exists a continuous and Lispschitz
neural operator $\hat {\cal K}_1$ such that, for all $x\in[0,1]$,
\begin{equation}
\left| {\cal K}_1(f,g)(x)-\hat {\cal K}_1(f,g)(x) \right| < \epsilon
\end{equation}
holds for all Lipschitz $f$ and $g$ with the properties that $\Vert f \Vert_{\infty}\leq B_f$ and $\Vert g \Vert_{\infty}  \leq B_g$.
\end{myth}

\subsection{Stabilization of  1-D hyperbolic equations under DeepONet gain feedback}
The following theorem states our main results regarding the stabilization properties of the backstepping design when the feedback gain is approximated by a DeepONet.
\begin{myth}
 Let $B_f > 0$, $B_g>0$ and $c> 0$ be arbitrarily large and consider the system (\ref{hyp1})--(\ref{hyp2}) for any $f\in{\cal C}^0\left({\cal T}\right)$ and $g\in{\cal C}^0\left([0,1]\right)$, both Lipschitz functions, which satisfy $\Vert g \Vert_\infty \leq B_g$ and $\Vert f \Vert_{\infty}\leq B_f$.
% There exists a sufficiently small $\epsilon(B_\lambda) >0$ such that t
The feedback 
\begin{equation}
U(t)=\int_0^1 \hat K_1(\xi)u(\xi,t) d\xi    
\end{equation}
with all NO gain kernels $\hat K_1 = \hat {\cal K}_1(f,g)$ of any approximation accuracy 
\begin{equation}
    0< \epsilon < \epsilon^*(B_f,B_g,c):= \frac{\sqrt{\frac{c }{2\mathrm{e}^{c}}}}{1+\left(B_f+ B_g\right)\mathrm{e}^{B_f+ B_g}}
\end{equation}
%$\epsilon \in (0, \epsilon^*)$ 
in relation to the exact backstepping kernel gain $ K_1 =  {\cal K}_1(f,g)$
%$K(1,x)(\lambda,c)$
 ensures that the closed-loop system satisfies the following
$L^2$ exponential stability bound for some $M>0$ with  decay rate given by $c/8$:
\begin{equation}
\Vert u(\cdot,t) \Vert_{L^2} \leq M \mathrm{e}^{-\frac{c}{8}(t-t_0)} \Vert u(\cdot,t_0) \Vert_{L^2}
\end{equation}
\end{myth}

\vspace{3pt}
 \begin{proof} Let $B_f > 0$, $B_g>0$ and $c> 0$ be arbitrarily large. Considering the use of the feedback law  $U(t)=\int_0^1 \hat K_1(\xi)u(\xi,t) d\xi $ in (\ref{eqn-targethypbc}) we reach
\begin{eqnarray}
w_t&=&w_x,\\
w(1,t)&=&\int_0^1 \left(\hat K_1(\xi)-K_1(\xi)\right) u(\xi,t) d\xi,
\end{eqnarray}
Now, using the exact inverse backstepping transformation,
\begin{eqnarray}
w_t&=&w_x,\\
w(1,t)&=&\int_0^1 \left(\hat K_1(\xi)-K_1(\xi)\right) \left[ w(\xi,t)+
% \nonumber \\  && \left.
\int_0^\xi L(\xi,s)w(s,t) ds \right] d\xi,
\end{eqnarray}
and switching the order of integration in the second part of the integral, and calling 
\begin{equation}
%G(\xi)= \left(\hat K(1,\xi)-K(1,\xi)\right) + \int^1_\xi L(s,\xi) \left(\hat K(1,s)-K(1,s)\right)ds\,,    
G(\xi)= -\tilde K_1(\xi) - \int^1_\xi L(s,\xi) \tilde K_1(s)ds\,,    
\end{equation}
where
    $\tilde K_1(\xi) :=  K_1(\xi)-\hat K_1(\xi)$.
 We reach:
\begin{eqnarray}
w_t&=&w_x,\\
w(1,t)&=&\int_0^1 G(\xi) w(\xi,t),
\end{eqnarray}

From Theorem~\ref{th-approx}, given $B_g$ and $B_f$, and for any $\epsilon>0$ such that $\epsilon<\epsilon^*$ with $\epsilon^*$ given in the Theorem statement, there exists a neural operator that ensures that
\begin{equation}
    |\tilde K_1(\xi)| < \epsilon\,, \quad \forall \xi\in[0,1]
\end{equation}
and therefore $
    \vert G(\xi) \vert \leq \epsilon (1+\| L \|_\infty)\, \forall \xi\in[0,1]$
where $\Vert L \Vert_\infty$ also depends on $B_f$ and $B_g$ as stated in Theorem~\ref{th-kernelboundhyp}. Now define
\begin{equation}V=\int_0^1 \mathrm{e}^{cx} w^2(x,t) dx\end{equation}
One then obtains
\begin{eqnarray}
\dot V&=&
%2\int_0^1 \mathrm{e}^{cx} w(x,t) w_t(x,t) dx
%\nonumber \\ &=&
2 \int_0^1 \mathrm{e}^{cx} w(x,t) w_x(x,t) dx
\end{eqnarray}

Note that 
\begin{eqnarray*}\int_0^1 \mathrm{e}^{cx} w(x,t) w_x(x,t) dx&=&-\frac{c}{2}\int_0^1 \mathrm{e}^{cx} w^2(x,t)
%\nonumber \\ &&
+ \frac{\mathrm{e}^{c}}{2} w^2(1,t)-\frac{1}{2}w^2(0,t)\end{eqnarray*}

Thus, defining  $\Vert G\Vert_\infty=\max_{\xi\in[0,1]} \vert G(\xi)\vert$,
\begin{eqnarray}
\dot V &\leq& -\frac{c}{2}\int_0^1 \mathrm{e}^{cx} w^2(x,t) + \frac{\mathrm{e}^{c}}{2}\left( \int_0^1 G(\xi) w(\xi,t) d\xi\right)^2 \nonumber \\
&\leq& -\frac{c}{2}\int_0^1 \mathrm{e}^{cx} w^2(x,t) + \frac{\mathrm{e}^{c}\Vert G\Vert_\infty^2}{2} \int_0^1  w^2(\xi,t) d\xi \normalcolor
\nonumber \\
&\leq&
-\left(\frac{c}{2}-\frac{\mathrm{e}^{c}\Vert G\Vert_\infty^2}{2} \right) V
\end{eqnarray}
and since
\begin{eqnarray}
\Vert G\Vert_\infty &\leq& \epsilon (1+\| L \|_\infty)
\nonumber \\&
<&  \epsilon^*  \left(1+(B_f+B_g)\mathrm{e}^{B_f+B_g}\right)
 \\&
\leq &\sqrt{\frac{c \mathrm{e}^{-c}}{2}}
\end{eqnarray}
and thus  we reach
$    \dot V \leq -\frac{c}{4} V
$
and the proof follows by the equivalence of $V$ to the square of the $L^2$ norm of $w$ and the use of the direct and inverse backstepping transformations, and their bounds, to express the obtained result in terms of the $L^2$ of $u$ using the bounds of Theorem~\ref{th-kernelboundhyp}, see e.g.~\cite{bhan2023neural}.
\end{proof}

\section{Reaction-Diffusion equation}\label{sect-RD}
Consider the plant
\begin{eqnarray}
u_t&=&u_{xx}+\lambda(x) u \label{eq-plant},
\end{eqnarray}
where $u(x,t)$ is the state, $x \in(0,1)$, $t>0$, for $\lambda\in{\cal C}^0\left([0,1]\right)$, with two possible boundary conditions (Dirichlet and Neumann cases). In the Dirichlet case we have
\begin{eqnarray}
u(0,t)&=&0,\quad
u(1,t)=U(t),\label{eq-bc2}
\end{eqnarray}
with $U(t)$ being the actuation, whereas in the Neumann case we have
\begin{eqnarray}
u_x(0,t)&=&0,\quad
u_x(1,t)=U(t).\label{eq-bc22n}
\end{eqnarray}

\subsection{Backstepping feedback law design for reaction-diffusion equations}
As in the hyperbolic 1-D case, we employ a direct/inverse backstepping transformation pair defined exactly as (\ref{eqn-hyptrans})--(\ref{eqn-hypinvtrans}). In this case, choosing some value of $c\geq 0$, the kernel equations verified by the direct transformation kernel~\cite{smyshlyaev2004closed} involving the coefficient  $\lambda$ of the system is as follows
\begin{eqnarray}
K_{xx}(x,\xi)-K_{\xi\xi}(x,\xi)&=&(\lambda(\xi)+c)K(x,\xi),\label{eqn-hypk1rd}\\
K(x,x)&=&-\frac{1}{2} \int_0^x \left(\lambda(s)+c\right)ds,\label{eqn-hypk2rd}
\end{eqnarray}
in the domain ${\cal T}$,
with the additional boundary condition $K(x,0)=0$ in the Dirichlet case and $K_{\xi}(x,0)=0$ in the Neumann case. A very similar equation is satisfied by $L(x,\xi)$. 
\subsubsection{Dirichlet case}
Applying (\ref{eqn-hyptrans}) to (\ref{eq-plant}), (\ref{eq-bc2}), with the kernel satisfying (\ref{eqn-hypk1rd})--(\ref{eqn-hypk2rd}) and $K(x,0)=0$ condition, one can find that the PDE verified by the new $w(x,t)$ variable, the target system, is
\begin{eqnarray}
w_t&=&w_{xx}-c w,\label{eqn-rdtarget1}\\
w(0,t)&=&0,\\
w(1,t)&=&U(t)-\int_0^1 K(1,\xi)u(\xi,t) d\xi,\label{eqn-rdtarget3}
\end{eqnarray}
Thus, defining the feedback gain in terms of the backstepping kernel as in the hyperbolic 1-D case,
$K_1^D(\xi)=K(1,\xi),    $
and using the control law 
\begin{equation}
    U(t)=\int_0^1 K_1^D(\xi)u(\xi,t) d\xi\,, 
\end{equation}
one achieves a homogeneous boundary condition $w(1,t)=0$ in (\ref{eqn-rdtarget3}) and, consequently,  the target system becomes exponentially stable. Since the backstepping transformation is invertible, this implies exponential stability in the original plant coordinates $u(x,t)$, see~\cite{smyshlyaev2004closed} for details.
\subsubsection{Neumann case}
As in the Dirichlet case, applying (\ref{eqn-hyptrans}) to (\ref{eq-plant}), (\ref{eq-bc22n}), with the kernel satisfying (\ref{eqn-hypk1rd})--(\ref{eqn-hypk2rd}) and $K_\xi(x,0)=0$ condition, one can find that the PDE verified by the new $w(x,t)$ variable, the target system, is
\begin{eqnarray}
w_t&=&w_{xx}-c w,\\
w_x(0,t)&=&0,\\
w_x(1,t)&=&-qw(1,t)+U(t)-(K(1,1)+q)u(1,t)
\nonumber \\
&&
-\int_0^1\left(K_x(1,\xi)-qK(1,\xi) \right)u(\xi,t) d\xi, \label{eqn-rdtarget4}
\end{eqnarray}
where the term $-qw(1,t)$, where $q>0$ can take any value,  adds some slight extra complications compared with the Dirichlet case. It has been added due to the reaction-diffusion equation with Neumann boundary conditions possessing a zero eigenvalue.
By using $U(t)=(K(1,1)+q)u(1,t)+\int_0^1\left(K_x(1,\xi)-qK(1,\xi) \right) u(\xi,t) d\xi$, the target system becomes a stable reaction-diffusion equation if $q>0$.

Noting that $K(1,1)$ is directly obtained from (\ref{eqn-hypk2rd}), denote then the gain as
$
K_1^N(\xi)=K_x(1,\xi)-qK(1,\xi)$,
and using the control law 
\begin{equation}
    U(t)=(K(1,1)+q)u(1,t)+\int_0^1 K_1^N(\xi)u(\xi,t) d\xi, 
\end{equation}
one achieves a  boundary condition $w(1,t)=-qw(1,t)$ in (\ref{eqn-rdtarget4}) and, consequently,  the target system becomes exponentially stable, see~\cite{smyshlyaev2004closed}.
\subsubsection{Kernel bounds for Dirichlet and Neumann cases}
In both the Dirichlet and Neumann cases, the next result follows~\cite{smyshlyaev2004closed}
\begin{myth} \label{th-kernelboundRD1}
Consider the equations verified by $K(x,\xi)$ (given by (\ref{eqn-hypk1rd})--(\ref{eqn-hypk2rd}) and $K(x,0)=0$ in the Dirichlet case or $K_\xi(x,0)=0$ in the Neumann case) and $L(x,\xi)$ (given in~\cite{smyshlyaev2004closed}) in the domain ${\cal T}$ with $\lambda\in{\cal C}^0\left([0,1]\right)$ and $c>0$.  Then, there exists a unique solution $K,L \in {\cal C}^1\left({\cal T}\right)$, and denoting $\bar \lambda=\Vert \lambda \Vert_{\infty}=\max_{x\in[0,1]} \vert \lambda(x)\vert$, one has
\begin{eqnarray}
\Vert L \Vert_{\infty},\Vert K \Vert_{\infty} &\leq&  \left(c+\bar \lambda\right) \mathrm{e}^{2\left(c+\bar \lambda\right)} \quad\mathrm{(Dirichlet)}\\
\Vert L \Vert_{\infty},\Vert K \Vert_{\infty} &\leq&  2\left(c+\bar \lambda\right) \mathrm{e}^{4\left(c+\bar \lambda\right)} \,\,\mathrm{(Neumann)}
\end{eqnarray}
\end{myth}
\subsection{Accuracy of approximation of backstepping reaction-diffusion gain operators with DeepONet}
As in Section~\ref{sect-accuracy}, we approximate (for both Dirichlet and Neumann cases) the gain operator $K_1(\xi)$ by DeepONet as $\hat K_1(\xi)$ as defined for each case.
Let the operators ${\cal K}_1^D:{\cal C}^0\left([0,1]\right)\times \mathbb R^{+}\rightarrow {\cal C}^1\left([0,1]\right)$ and ${\cal K}_1^N:{\cal C}^0\left([0,1]\right)\times \mathbb R^{+}\times \mathbb R^{+}\rightarrow {\cal C}^1\left([0,1]\right)$ be given by
\begin{equation}
K_1^D(x)=: {\cal K}_1^D(\lambda,c)(x),\quad K_1^N(x)=: {\cal K}_1^N(\lambda,c,q)(x)
\end{equation}
As in Theorem~\ref{th-approx} for the hyperbolic case, we get the following key result for the
approximation of the reaction-diffusion backstepping kernel gains by a DeepONet, which we state simultaneously both for the  ${\cal K}_1^{D}$ and ${\cal K}_1^{N}$ operators (Dirichlet and Neumann cases).
\begin{myth}\label{th-approxRD}
For all $B_\lambda$, $c> 0$ and $\epsilon > 0$ (and $q>0$ in the Neumann case), there exists a continuous and Lispschitz
neural operator $\hat {\cal K}_1^D$ (resp. $\hat {\cal K}_1^N$ in the Neumann case) such that, for all $x\in[0,1]$, the following holds for all Lipschitz $\lambda$ with the property that $\Vert \lambda \Vert_{\infty}\leq B_\lambda$:
\begin{equation}
\left| {\cal K}_1^D(\lambda,c)(x)-\hat {\cal K}_1^D(\lambda,c)(x) \right| < \epsilon
\end{equation}
in the Dirichlet case or
\begin{equation}
\left| {\cal K}_1^N(\lambda,c,q)(x)-\hat {\cal K}_1^N(\lambda,c,q)(x) \right| < \epsilon
\end{equation}
in the Neumann case.
\end{myth}

\subsection{Stabilization of reaction-difussion equations under DeepONet gain feedback}
Next, we state our main stability results  when the backstepping feedback gain is approximated by a DeepONet, both in the Dirichlet and Neumann cases, which are given separately due to substantial differences.
\subsubsection{Dirichlet case}

In the Dirichlet case, one obtains an $H^1$ stabilization result, as given next.
\begin{myth}\label{th-stabRDD}
 Let $B_\lambda > 0$ and $c\geq 0$ be arbitrarily large and consider the system (\ref{eq-plant})--(\ref{eq-bc2}) for any $\lambda \in {\cal C}^0([0,1])$ a Lipschitz function which satisfies $\Vert \lambda \Vert_\infty \leq B_\lambda$. 
% There exists a sufficiently small $\epsilon(B_\lambda) >0$ such that t
The feedback 
\begin{equation}
U(t)=\int_0^1 \hat K_1(\xi)u(\xi,t) d\xi    
\end{equation}
with all NO gain kernels $\hat K_1^D = \hat {\cal K}_1^D(c,\lambda)$ of any approximation accuracy 
\begin{equation}
    0< \epsilon < \epsilon^*(B_\lambda,c)
    :=\frac{1}{\sqrt{20}\left(1+\left(c+B_\lambda\right)\mathrm{e}^{2\left(c+B_\lambda\right)}\right)}
\end{equation}
%$\epsilon \in (0, \epsilon^*)$ 
in relation to the exact backstepping kernel gain $ K_1^D =  {\cal K}_1^D(c,\lambda)$
%$K(1,x)(\lambda,c)$
 ensures that the closed-loop system satisfies the following
$H^1$ exponential stability bound with arbitrary decay rate:
\begin{equation}
\Vert u(\cdot,t) \Vert_{H^1} \leq M \mathrm{e}^{-\left(c+\frac{1}{12}\right)(t-t_0)} \Vert u(\cdot,t_0) \Vert_{H^1}
\end{equation}
\end{myth}

%See the ArXiv preprint of this paper~\cite{vazquez2024gain} for the proof.
 \begin{proof}
 Let $B_\lambda>0$ and $c> 0$ be arbitrarily large. Considering the use of the feedback law  $U(t)=\int_0^1 \hat K_1^D(\xi)u(\xi,t) d\xi $ in (\ref{eqn-rdtarget3}) we reach
 \begin{eqnarray}
 w_t&=&w_{xx}-c w,\\
 w(0,t)&=&0,\\
 w(1,t)&=&\int_0^1 \left(\hat K_1^D(\xi)-K_1^D(\xi)\right) u(\xi,t) d\xi,
 \end{eqnarray}

 To write everything as a function of $w$, let us use the exact inverse:
 \begin{eqnarray}
 w_t&=&w_{xx}-c w,\\
 w(0,t)&=&0,\\
 w(1,t)&=&\int_0^1 \left(\hat K_1^D(\xi)-K_1^D(\xi)\right) \left[ w(\xi,t)
% \nonumber \\ && \left.
 +\int_0^\xi L(\xi,s)w(s,t) ds \right] d\xi,
 \end{eqnarray}
 switching the order of integration in the second part of the integral, and calling 
 \begin{equation}
 %G(\xi)= \left(\hat K(1,\xi)-K(1,\xi)\right) + \int^1_\xi L(s,\xi) \left(\hat K(1,s)-K(1,s)\right)ds\,,    
 G^D(\xi)= -\tilde K_1^D(\xi) - \int^1_\xi L(s,\xi) \tilde K_1^D(s)ds\,,    
 \end{equation}
 where
 \begin{equation}
     \tilde K_1^D(\xi) :=  K_1^D(\xi)-\hat K_1^D(\xi)\,,
 \end{equation}
 we reach:
 \begin{eqnarray}
 w_t&=&w_{xx}-c w,\\
 w(0,t)&=&0,\\
 w(1,t)&=&\int_0^1 G^D(\xi) w(\xi,t),
 \end{eqnarray}

 From Theorem~\ref{th-approxRD}, given $B_\lambda$ and $c$, and for any $\epsilon>0$ such that $\epsilon<\epsilon^*$ with $\epsilon^*$ given in the Theorem~\ref{th-stabRDD} statement, there exists a neural operator that ensures that
 \begin{equation}
     |\tilde K_1^D(\xi)| < \epsilon\,, \quad \forall \xi\in[0,1]
 \end{equation}
 and therefore% seek to have
 \begin{equation}
     \vert G^D(\xi) \vert \leq \epsilon (1+\| L \|_\infty)\,, \quad \forall \xi\in[0,1]
 \end{equation}
 where $\Vert L \Vert_\infty$ also depends on $B_\lambda$ and $c$ as stated in Theorem~\ref{th-kernelboundRD1}.

 Define the following Lyapunov functionals:
 \begin{eqnarray}
 V_1&=&\frac{1}{2} \int_0^1 w^2 (x,t) dx,\\
 V_2&=&\frac{1}{2} \int_0^1 w_x^2 (x,t) dx,
 \end{eqnarray}
 we get
 \begin{eqnarray}
  \dot V_1&=&-2c V_1- 2V_2+w_x(1,t) \int_0^1 G^D(\xi) w(\xi,t) ,\\
 \dot V_2&=& \int_0^1 w_x w_{xt} (x,t) dx 
 \nonumber \\&=& -  \int_0^1 w_{xx} w_{t} (x,t) dx + w_x(1,t) \int_0^1 G^D(\xi) w_t(\xi,t) \nonumber \\
 &=& - \int_0^1 w_{xx}^2 (x,t) dx  +c  \int_0^1 w_{xx} w(x,t) dx 
% \nonumber \\ &&
 +w_x(1,t) \int_0^1 G^D(\xi) w_{xx}(\xi,t)d\xi
 \nonumber \\ &&
 -cw_x(1,t) \int_0^1 G^D(\xi) w(\xi,t)d\xi \nonumber \\
 &=& - \int_0^1 w_{xx}^2 (x,t) dx  -2c V_2   
 %\nonumber \\ &&
 +w_x(1,t) \int_0^1 G^D(\xi) w_{xx}(\xi,t)
 \end{eqnarray}
 Now we can use the fact that
 \begin{equation}
 f(1)=\int_0^1 \left[\frac{d}{dx} xf(x)\right] dx =\int_0^1 f(x) dx + \int_0^1 xf'(x) dx
 \end{equation}
 which squared and applying several inequalities gives
 \begin{equation}
 f^2(1) \leq 2 \left( \int_0^1 f^2(x) dx + \int_0^1 f_x^2(x) dx \right)
 \end{equation}
 Thus we get the classical bound for a trace:
 \begin{equation} w^2_x(1,t) \leq 2 \left(\int_0^1 w_{xx}^2 (x,t) dx+\int_0^1 w_{x}^2 (x,t) dx\right)
 \end{equation}
 Consider $V=V_1+\alpha V_2$, equivalent to the $H^1$ norm, with $\alpha$ to be set later. We get
 \begin{eqnarray}
  \dot V&=&-2c V_1- 2V_2+w_x(1,t) \int_0^1 G^D(\xi) w(\xi,t)  
%  \nonumber \\ &&
 +\alpha \left[ - \int_0^1 w_{xx}^2 (x,t) dx  -2c V_2  
 \right.
 \nonumber \\ && \left.
 +w_x(1,t) \int_0^1 G^D(\xi) w_{xx}(\xi,t)\right]
 \nonumber \\
 &\leq&-2c V_1- 2V_2 +\frac{w_x^2(1,t)}{2\delta_1} + \delta_1 \Vert G^D \Vert_\infty^2 V_1
 \nonumber \\
 &&+\alpha \left[ - \int_0^1 w_{xx}^2 (x,t) dx  -2c V_2+\frac{w_x^2(1,t)}{2\delta_2}
 %\right.
 %\nonumber \\ && \left.
 +  \frac{\delta_2 \Vert G^D \Vert_\infty^2}{2} \int_0^1 w_{xx}^2 (x,t) dx  \right] \normalcolor
 \end{eqnarray}
 Choosing $\delta_2=\delta_1=\frac{1}{2\Vert G^D \Vert_\infty^2}$:
 \begin{eqnarray}
  \dot V\hspace{-3pt}
 &\hspace{-3pt}\leq\hspace{-3pt}&-\left(2c-\frac{1}{2}\right) V_1- 2V_2 +\Vert G^D \Vert_\infty^2 w_x^2(1,t)
 \nonumber \\
 &&
 +\alpha\left[ - \frac{3}{4} \int_0^1 w_{xx}^2 (x,t) dx  -2c V_2+\Vert G^D \Vert_\infty^2 w_x^2(1,t)\right]
 \nonumber 
 \end{eqnarray}
 and using the bound on $w_x(1,t)$:
 \begin{eqnarray}
  \dot V\hspace{-3pt}
 &\hspace{-3pt}\leq\hspace{-3pt}&-\left(2c-\frac{1}{2}\right) V_1- 2(1-\Vert G^D \Vert_\infty^2) V_2 
% \nonumber \\ && 
 +2\Vert G^D \Vert_\infty^2 \int_0^1 w_{xx}^2 (x,t) dx
 \nonumber \\ && 
 +\alpha \left[
 -2(c-\Vert G^D \Vert_\infty^2) V_2
 %\nonumber \\
 %&&\left.
 -\left( \frac{3}{4}-2\Vert G^D \Vert_\infty^2\right) \int_0^1 w_{xx}^2 (x,t) dx 
 \right]
 \nonumber \\
 &=&
 -\left(2c-\frac{1}{2}\right) V_1- 2\left(1-\Vert G^D \Vert_\infty^2
% \right. \nonumber \\ &&  \left.
 +2\alpha(c-\Vert G^D \Vert_\infty^2)\right) V_2
 \nonumber \\ &&
 -\left(\alpha \left[  \frac{3}{4}-2\Vert G^D \Vert_\infty^2\right]
 %\right.   \left.
 -2\Vert G^D \Vert_\infty^2 \right) \int_0^1 w_{xx}^2 (x,t) dx
 \end{eqnarray}
 Choosing $\alpha=4$,
 \begin{eqnarray}
  \dot V
 &\leq&-\left(2c-\frac{1}{2}\right) V_1- \left(2(1+4c)-10\Vert G^D\Vert^2 \right) V_2
 \nonumber \\
 &&
 -\left(3-10\Vert G^D \Vert_\infty^2\right) \int_0^1 w_{xx}^2 (x,t) dx
 \end{eqnarray}
 and since
 \begin{eqnarray}
 \Vert G^D\Vert_\infty &\leq& \epsilon (1+\| L \|_\infty)
 \nonumber \\&
 <&  \epsilon^*  \left(1+(B_\lambda+c)\mathrm{e}^{2(B_\lambda+c)}\right)
  \\&
 \leq &\sqrt{\frac{1}{20}},
 \end{eqnarray}
 we get
 \begin{equation}
  \dot V
 \leq -\left(2c-\frac{1}{2}\right) V_1- \left(2c+\frac{3}{2} \right) V_2  \label{eq-lyap0}
 \end{equation}
 By Poincare's inequality it holds that $V_1\leq 2V_2$, thus
 \begin{equation}
  \dot V
 \leq -\left(2c+\frac{1}{6} \right) V_1- \left(2c+\frac{1}{6} \right) V_2=- \left(2c+\frac{1}{6} \right) V \label{eq-lyap}
 \end{equation}
 and thus we obtain $H^1$ stability with convergence rate $c+\frac{1}{12}$ for $w$ by using the norm equivalence of $V$ to the $H^1$ norm (squared). Finishing proof is straightforward employing the direct and inverse transformations (and their derivatives since we are considering the $H^1$ norm) to obtain equivalence between the $H^1$ norms of $w$ and $u$.
 \end{proof}

Note that, for (\ref{eq-plant})--(\ref{eq-bc2}) to have a solution in the $H^1$ space would require, in principle, compatibility conditions and initial conditions in $H^1$, but the smoothing effect of the parabolic operator should ensure this is the case in an infinitesimal time after possibly a small jump. For simplicity this argument is skipped and the Theorem implicitely assumes that $H^1$ solutions do indeed exist for (\ref{eq-plant})--(\ref{eq-bc2}).

\subsubsection{Neumann case}

In the Neumann case, one obtains an $L^2$ stabilization result, as given next.

\begin{myth}
 Let $B_\lambda > 0$, $q>1$ and $c\geq 0$ be arbitrarily large and consider the system (\ref{eq-plant})--(\ref{eq-bc22n}) for any $\lambda \in {\cal C}^0([0,1])$ a Lipschitz function which satisfies $\Vert \lambda \Vert_\infty \leq B_\lambda$. 
% There exists a sufficiently small $\epsilon(B_\lambda) >0$ such that t
The feedback 
\begin{equation}
U(t)=(K(1,1)+q)u(1,t)+\int_0^1 \hat K_1^N(\xi)u(\xi,t) d\xi  
\end{equation}
with all NO gain kernels $\hat K_1^N = \hat {\cal K}_1^N(c,\lambda,q)$ of any approximation accuracy 
\begin{equation}
    0< \epsilon < \epsilon^*(B_\lambda,c,q)
    :=\sqrt{\frac{q-1}{2}}\frac{1}{\left(1+2\left(c+B_\lambda\right)\mathrm{e}^{4\left(c+B_\lambda\right)}\right)}
\end{equation}
%$\epsilon \in (0, \epsilon^*)$ 
in relation to the exact backstepping kernel gain $ K_1^N =  {\cal K}_1^N(c,\lambda,q)$
%$K(1,x)(\lambda,c)$
 ensures  the closed-loop system satisfies the following
$L^2$ exponential stability bound with arbitrary decay:
\begin{equation}
\Vert u(\cdot,t) \Vert_{L^2} \leq M \mathrm{e}^{-\left(c+\frac{1}{8}\right)(t-t_0)} \Vert u(\cdot,t_0) \Vert_{L^2}
\end{equation}
\end{myth}
%See the ArXiv preprint of this paper~\cite{vazquez2024gain} for the proof.
 \begin{proof}
 Let $B_\lambda>0$, $c> 0$ and $q>0$ be arbitrarily large. Considering the use of the feedback law $U(t)=(K(1,1)+q)u(1,t)+\int_0^1 \hat K_1(\xi)u(\xi,t) d\xi$ in (\ref{eqn-rdtarget4}) we reach
 \begin{eqnarray}
 w_t&=&w_{xx}-c w,\\
 w_x(0,t)&=&0,\\
 w_x(1,t)&=&-qw(1,t)
% \nonumber \\ &&
 +\int_0^1 \left(\hat K_1^N(\xi)-K_1^N(\xi)\right) u(\xi,t) d\xi,
 \end{eqnarray}
 As in the Dirichlet case, to write everything as a function of $w$, let us use the exact inverse defining 
 \begin{equation}
 G^N(\xi)= -\tilde K_1^N(\xi) - \int^1_\xi L(s,\xi) \tilde K_1^N(s)ds\,,    
 \end{equation}
 where
 \begin{equation}
     \tilde K_1^N(\xi) :=  K_1^N(\xi)-\hat K_1^N\xi)\,,
 \end{equation}
 thus reaching
 \begin{eqnarray}
 w_t&=&w_{xx}-c w,\\
 w_x(0,t)&=&0,\\
 w_x(1,t)&=&-qw(1,t)+\int_0^1 G^N(\xi) w(\xi,t),
 \end{eqnarray}
 From Theorem~\ref{th-approxRD}, given $B_\lambda$ and $c$ positive, and for any $\epsilon>0$ such that $\epsilon<\epsilon^*$ with $\epsilon^*$ given in the Theorem~\ref{th-stabRDD} statement, there exists a neural operator that ensures that
 \begin{equation}
     |\tilde K_1^N(\xi)| < \epsilon\,, \quad \forall \xi\in[0,1]
 \end{equation}
 and therefore% seek to have
 \begin{equation}
     \vert G^N(\xi) \vert \leq \epsilon (1+\| L \|_\infty)\,, \quad \forall \xi\in[0,1]
 \end{equation}
 where $\Vert L \Vert_\infty$ also depends on $B_\lambda$ and $c$ as stated in Theorem~\ref{th-kernelboundRD1} (Neumann case).
 Define the following Lyapunov functional
 \begin{eqnarray}
 V_1&=&\frac{1}{2} \int_0^1 w^2 (x,t) dx,
 \end{eqnarray}
 then we get
 \begin{eqnarray}
  \dot V_1&=&-2c V_1- \int_0^1w_x^2-q w^2(1,t)
%  \nonumber \\  &&
  +w(1,t) \int_0^1 G^N(\xi) w(\xi,t) 
 \end{eqnarray}
 Now, Poincare's inequality tells us that $\int_0^1 w^2 (x,tdx)\leq 2w^2(1,t) +2\int_0^1 w_x^2 (x,t) dx $. We get
 \begin{eqnarray}
  \dot V_1&=&-2c V_1- \int_0^1w_x^2-q w^2(1,t)
%  \nonumber \\  &&
  +w(1,t) \int_0^1 G^N(\xi) w(\xi,t)  \nonumber \\
 &\leq& -2c V_1- \int_0^1w_x^2-q w^2(1,t) +\frac{w^2(1,t)}{2\delta_1} 
% \nonumber \\ &&
 + \delta_1 \Vert G^N \Vert_\infty^2 V_1
 \nonumber \\
 &\leq&-\left(2c+\frac{1}{2}-\delta_1 \Vert G^N \Vert_\infty^2\right) V_1
% \nonumber \\  &&
 -\left(q-1-\frac{1}{2\delta_1}\right) w^2(1,t) 
 \end{eqnarray}
 Choosing $\delta_1=\frac{1}{4\Vert G^N \Vert_\infty^2}$:
 \begin{eqnarray*}
  \dot V_1
 &\leq&-\left(2c+\frac{1}{4}\right) V_1-\left(q-1-2\Vert G^N \Vert_\infty^2 \right) w^2(1,t) 
 \end{eqnarray*}
 and since
 \begin{eqnarray}
 \Vert G^N\Vert_\infty &\leq& \epsilon (1+\| L \|_\infty)
 \nonumber \\&
 <&  \epsilon^*  \left(1+2(B_\lambda+c)\mathrm{e}^{4(B_\lambda+c)}\right)
  \nonumber \\&
 \leq &\sqrt{\frac{q-1}{2}},
 \end{eqnarray}
 we get
 \begin{equation}
  \dot V_1
 \leq -\left(2c+\frac{1}{4}\right) V_1  \label{eq-lyap2}
 \end{equation}
 and thus we obtain $H^1$ stability with convergence rate $c+\frac{1}{8}$ for $w$ by using the norm equivalence of $V$ to the $L^2$ norm (squared). Finishing the proof is straightforward employing the direct and inverse transformations  to obtain equivalence between the $L^2$ norms of $w$ and $u$.
 \end{proof}

\section{Hyperbolic $2 \times 2$ system}\label{sect-hypsist}

Consider the following hyperbolic \(2 \times 2\) system:
\begin{eqnarray}
 \partial_t u(x,t) &=& -\lambda(x) u_x(x,t) + \sigma(x) u(x,t) + \omega(x) v(x,t), \label{eq-plant2h1} \\
 \partial_t v(x,t) &=& \mu(x) v_x(x,t) + \theta(x) u(x,t), \label{eq-plant2h2}
\end{eqnarray}
for \( x \in [0,1] \), \( t > 0 \), where \( u(x,t) \) and \( v(x,t) \) are the system states. The boundary conditions are given by
\begin{eqnarray}
u(0,t) &=& q v(0,t), \label{eq-bc2h1} \\
v(1,t) &=& U(t), \label{eq-bc2h2}
\end{eqnarray}
with \( U(t) \) being the control input and \( q \neq 0 \) a given constant. We assume the following properties on the system coefficients.  The functions \( \lambda, \mu \in \mathcal{C}^1([0,1]) \) are positive and bounded away from zero:
    \begin{equation}
    0 < C_\lambda \leq \lambda(x), \quad 0 < C_\mu \leq \mu(x)  \quad \forall x \in [0,1],
    \end{equation}
    where \( C_\lambda, C_\mu \) are known constants. In addition, \( \sigma, \omega, \theta \in \mathcal{C}^0([0,1]) \).

\subsection{Backstepping Feedback Law Design for Hyperbolic \(2 \times 2\) Systems}

To stabilize the system \eqref{eq-plant2h1}--\eqref{eq-bc2h2}, following  \cite{vazquez-nonlinear}, we employ the backstepping method with the following direct and inverse transformations:
\begin{eqnarray}
\beta(x,t) &=& v(x,t) - \int_0^x k_u(x,\xi) u(\xi,t) d\xi - \int_0^x k_v(x,\xi) v(\xi,t) d\xi, \label{eq-hyp2-trans1} \\
v(x,t) &=& \beta(x,t) + \int_0^x l_u(x,\xi) u(\xi,t) d\xi + \int_0^x l_v(x,\xi) \beta(\xi,t) d\xi, \label{eq-hyp2-trans2}
\end{eqnarray}
where \( k_u(x,\xi) \), \( k_v(x,\xi) \) are the backstepping kernels, and \( l_u(x,\xi) \), \( l_v(x,\xi) \) are the inverse kernels. Notice that the $u(x,t)$ state is not transformed; only the controlled state, $v(x,t)$, is mapped into a target variable. The kernels \( k_u(x,\xi) \) and \( k_v(x,\xi) \) are solutions to the following first-order hyperbolic PDEs  defined on the triangular domain $\cal{T}$.

\begin{eqnarray}
\mu(x) \frac{\partial k_u(x,\xi)}{\partial x} - \lambda (\xi) \frac{\partial k_u(x,\xi)}{\partial \xi} &=&\lambda'(\xi)k_u(x,\xi) +  \sigma(\xi) k_u(x,\xi) + \theta(\xi) k_v(x,\xi), \label{eq-ku} \\
\mu (x) \frac{\partial k_v(x,\xi)}{\partial x} + \mu(\xi) \frac{\partial k_v(x,\xi)}{\partial \xi} &=&\mu'(\xi)k_v(x,\xi)+  \omega(\xi) k_u(x,\xi). \label{eq-kv}
\end{eqnarray}
These equations are subject to the boundary conditions:
\begin{equation}
k_u(x,x) = -\frac{\theta(x)}{\lambda(x)+\mu(x)}, \quad k_v(x,0) =\frac{q \lambda(0)}{\mu(0)} k_u(x,0). \label{eq-kernel-bc}
\end{equation}
The existence and uniqueness of  solutions for  \( k_u(x,\xi) \) and \( k_v(x,\xi) \) are guaranteed under the given assumptions on the system parameters \cite[Appendix A, Theorem A.1]{vazquez-nonlinear}. The inverse backstepping kernels verify very similar equations (see the same paper) and thus existence and uniqueness of  solutions for  \( l_u(x,\xi) \) and \( l_v(x,\xi) \) are also guaranteed.

Applying the transformation \eqref{eq-hyp2-trans1} to the plant equations \eqref{eq-plant2h1}--\eqref{eq-plant2h2}, we obtain the target system:
\begin{eqnarray}
\partial_t u(x,t) &=& -\lambda(x) u_x(x,t) + \sigma(x) u(x,t) + \omega(x) \beta(x,t) \nonumber \\
&& + \int_0^x \kappa_u(x,\xi) u(\xi,t) d\xi + \int_0^x \kappa_v(x,\xi) \beta(\xi,t) d\xi, \label{eq-hyp2-target1} \\
\partial_t \beta(x,t) &=& \mu(x) \beta_x(x,t), \label{eq-hyp2-target2}
\end{eqnarray}
where 
\begin{equation} \kappa_u(x,\xi)=\omega(x)l_u(x,\xi),\quad \kappa_v(x,\xi)=\omega(x)l_v(x,\xi) \label{eq-kappa-def} \end{equation}
with boundary conditions
\begin{eqnarray}
u(0,t) &=& q \beta(0,t), \label{eq-hyp2-bc1} \\
\beta(1,t) &=& U(t) - \int_0^1 k_u(1,\xi) u(\xi,t) d\xi - \int_0^1 k_v(1,\xi) v(\xi,t) d\xi. \label{eq-hyp2-bc2}
\end{eqnarray}

Based in the shape of \eqref{eq-hyp2-bc2}, we design the control law:
\begin{equation}
U(t) = \int_0^1 k_{u1}(\xi) u(\xi,t) d\xi + \int_0^1 k_{v1}(\xi) v(\xi,t) d\xi, \label{eq-hyp2-control}
\end{equation}
where \( k_{u1}(\xi) = k_u(1,\xi) \) and \( k_{v1}(\xi) = k_v(1,\xi) \) are the controller gain functions.

With this control law, the boundary condition simplifies to \( \beta(1,t) = 0 \), and the target system \eqref{eq-hyp2-target1}--\eqref{eq-hyp2-target2} becomes exponentially stable.

\subsubsection{Kernel bounds}
Adapting the proofs in \cite[Appendix A, Theorem A.1 and Proposition A.6]{vazquez-nonlinear} the following result can be shown.
\begin{myth} \label{th-kernelbound-hyp2}
Consider the equations verified by \( k_u(x,\xi) \) and \( k_v(x,\xi) \) (given by  \eqref{eq-hyp2-target1}--\eqref{eq-kernel-bc}) and \(l_u(x,\xi) \) and \( l_v(x,\xi) \) (given in~\cite{vazquez-nonlinear}) in the domain ${\cal T}$ with
\( \lambda, \mu \in \mathcal{C}^1([0,1]) \)  positive and bounded away from zero, $q\neq 0$ and \( \sigma, \omega, \theta \in \mathcal{C}^0([0,1]) \). Then, there exists a unique solution $k_u(x,\xi)$, $k_v(x,\xi)$, $l_u(x,\xi)$, $l_v(x,\xi) \in {\cal C}^1\left({\cal T}\right)$, and denoting 
\begin{eqnarray*}
{C_\lambda}&=&\min_{x \in[0,1]} \vert \lambda(x) \vert, \,
\bar{\lambda'}=\max_{x \in[0,1]} \vert \lambda'(x) \vert, \,
{C_\mu}=\min_{x \in[0,1]} \vert \lambda(x) \vert, \,
% \nonumber \\
\bar{\mu'}=\max_{x \in[0,1]} \vert \mu'(x) \vert.
\end{eqnarray*}
and similarly
\begin{equation}
\bar{\sigma}=\max_{x \in[0,1]} \vert \sigma(x) \vert, \,
\bar{\omega}=\min_{x \in[0,1]} \vert \omega(x) \vert, \,
\bar{\theta}=\max_{x \in[0,1]} \vert \theta(x) \vert.
\end{equation}
 one has
\begin{eqnarray}
\Vert k_u \Vert_{\infty},\Vert k_v \Vert_{\infty},\Vert l_u \Vert_{\infty},\Vert l_v \Vert_{\infty} &\leq& K_1 \mathrm{e}^{K_2}
\end{eqnarray}
with $K_1=C_1 C_2$ and $K_2=C_2 C_3$, where $C_1=\bar{\theta} \frac{q \lambda(0)}{\mu(0)} $, $C_2=\max\left\{1/C_\lambda, 1/C_\mu\right\}$, $C_3=(1+q)(\bar{\lambda'}+\bar{\mu'}+\bar{\sigma}+\bar{\omega}+\bar{\theta})$.
\end{myth}

\subsection{Accuracy of Approximation of Backstepping Gain Operators with DeepONet}

As in the other cases, we aim to approximate the gain functions \( k_{u1}(\xi) \) and \( k_{v1}(\xi) \) using neural operators. Let the gain operators $\mathcal{K}_u, \mathcal{K}_v:\left({\cal C}^1\left([0,1\right)\right)^2 \times\left({\cal C}^0\left([0,1\right)\right)^3\times \mathbb{R}\rightarrow {\cal C}^1\left([0,1]\right)$ be defined as
\begin{equation}
k_{u1}(x) =: \mathcal{K}_u(\lambda, \mu, \sigma, \omega, \theta, q)(x), \quad k_{v1}(x) =: \mathcal{K}_v(\lambda, \mu, \sigma, \omega, \theta, q)(x).
\end{equation}
Applying as in the other sections the universal approximation theorem for neural operators~\cite{deng2022approximation}, we obtain:
\begin{myth}\label{th-approx-hyp2}
For all given bounds \( B_\lambda, B_\mu,B_{\lambda'}, B_{\mu'}, B_\sigma, B_\omega, B_\theta > 0 \), positive constants \( C_\lambda, C_\mu > 0 \), and any \( \epsilon > 0 \), there exist continuous and Lipschitz neural operators \( \hat{\mathcal{K}}_u \) and \( \hat{\mathcal{K}}_v \) such that, for all \( x \in [0,1] \),
\begin{equation}
\left| \mathcal{K}_u(\lambda, \mu, \sigma, \omega, \theta, q)(x) - \hat{\mathcal{K}}_u(\lambda, \mu, \sigma, \omega, \theta, q)(x) \right| < \epsilon,
\end{equation}
\begin{equation}
\left| \mathcal{K}_v(\lambda, \mu, \sigma, \omega, \theta, q)(x) - \hat{\mathcal{K}}_v(\lambda, \mu, \sigma, \omega, \theta, q)(x) \right| < \epsilon,
\end{equation}
for all Lipschitz continuous functions \( \lambda, \mu, \sigma, \omega, \theta \) satisfying \( \Vert \lambda \Vert_\infty \leq B_\lambda\),\(\Vert \lambda' \Vert_\infty \leq B_{\lambda'} \), \( \Vert \mu \Vert_\infty\leq  B_\mu\), \(\Vert \mu' \Vert_\infty\leq B_{\mu'} \), \( \Vert \sigma \Vert_\infty \leq B_\sigma \), \( \Vert \omega \Vert_\infty \leq B_\omega \), \( \Vert \theta \Vert_\infty \leq B_\theta \), and \( \lambda(x) \geq C_\lambda \), \( \mu(x) \geq C_\mu \).
\end{myth}

\subsection{Stabilization under DeepONet Gain Feedback for hyperbolic $2\times2$ systems}
The following theorem states our main results regarding the stabilization properties of the backstepping design when the feedback gain is approximated by a DeepONet for hyperbolic $2\times2$ systems, with arbitrary convergence rate.
\begin{myth}\label{th-stab-hyp2}
 Let $B_\lambda,B_\mu,B_\sigma,B_\omega,B_\theta,B_{\lambda'},B_{\mu'},B_{\lambda'},B_{\mu'}> 0$  be arbitrarily large and $0<C_{\lambda}<B_{\lambda}$, $0<C_{\mu}<B_{\mu}$ be arbitrarily small, and consider the system  (\ref{eq-plant2h1})--(\ref{eq-bc2h2}) for any $\lambda,\mu \in {\cal C}^1([0,1])$ with Lipschitz derivative and $\sigma,\omega,\theta \in {\cal C}^0([0,1])$ Lipschitz functions, all of which satisfy $\Vert \lambda \Vert_\infty \leq B_\lambda$, $\Vert \mu \Vert_\infty \leq B_\mu$,$\Vert \sigma \Vert_\infty \leq B_\sigma$,$\Vert \omega \Vert_\infty \leq B_\omega$,$\Vert \theta \Vert_\infty \leq B_\theta$,$\Vert \lambda' \Vert_\infty \leq B_{\lambda'}$,$\Vert \mu' \Vert_\infty \leq B_{\mu'}$ and $\lambda(x)>C_\lambda$, $\mu(x)>C_\mu$ $\forall x\in[0,1]$. Denote  $\breve B =(B_\lambda,B_\mu,B_\sigma,B_\omega,B_\theta,B_{\lambda'},B_{\mu'},C_{\lambda},C_{\mu},\bar c,q)$.  Define $K_1$ and $K_2$ from $\breve B$ as given in Theorem~\ref{th-kernelbound-hyp2}. Let $\bar c>0$ be any desired convergence rate. Then %There exists a sufficiently small $\epsilon(B_\lambda) >0$ such that t
the feedback 
\begin{equation}
U(t)=\int_0^1 \hat k_{u1}(\xi)u(\xi,t) d\xi+\int_0^x \hat k_{v1}(\xi)v(\xi,t) d\xi 
\end{equation}
with all NO gain kernels $\hat k_{u1} = \hat {\cal K}_{u}(\lambda,\mu,\sigma,\omega,\theta,q)$, $\hat k_{v1} = \hat {\cal K}_{v}(\lambda,\mu,\sigma,\omega,\theta,q)$,  of any approximation accuracy 
\begin{equation} \label{eq-epsilonstar-hyp}
    0< \epsilon < \epsilon^* (\breve B)
    %(B_\lambda,B_\mu,B_\sigma,B_\omega,B_\theta,B_{\lambda'},B_{\mu'},C_{\lambda},C_{\mu},\bar c,q)
    :=\frac{\sqrt{\frac{c(\breve B,\bar c)C_\lambda}{4M_\lambda}(1+q^2)\mathrm{e}^{-2c(\breve B,\bar c)}}+\sqrt{ \frac{c(\breve B,\bar c)C_\mu }{4 M_\mu}\mathrm{e}^{-c(\breve B,\bar c)}}}{\left(1+2K_1  \mathrm{e}^{K_2}\right)}
\end{equation}
%$\epsilon \in (0, \epsilon^*)$ 
in relation to the exact backstepping kernel gains $\ k_{u1} =  {\cal K}_{u}(\lambda,\mu,\sigma,\omega,\theta,q)$, $ k_{v1} =  {\cal K}_{v}(\lambda,\mu,\sigma,\omega,\theta,q)$
%$K(1,x)(\lambda,c)$
 ensures that the closed-loop system satisfies the following
$L^2$ exponential stability bound with  decay rate $\bar c/2$:
\begin{eqnarray}
&& \Vert u(\cdot,t) \Vert_{L^2} + \Vert v(\cdot,t) \Vert_{L^2} 
%\nonumber \\ 
\leq
M \mathrm{e}^{-\frac{\bar c}{2}(t-t_0)}\left( \Vert u(\cdot,t_0) \Vert_{L^2}+\Vert v(\cdot,t_0) \Vert_{L^2}\right),
\end{eqnarray}
where the number $c>0$ appearing in \eqref{eq-epsilonstar-hyp}
is defined as
\begin{eqnarray}
c &=& \frac{2\bar c}{\min \left\{C_\lambda,C_\mu \right\}} 
\nonumber \\ &&
+2 \max\left\{
\frac{1}{C_\lambda} \left(2 B_\sigma
+
B_\omega \left(1+
 K_1 \mathrm{e}^{K_2} \left(1+ \sqrt{\frac{B_\lambda}{C_{\lambda}}} \right)\right)\right),
\frac{B_\mu}{C_\lambda C_\mu}B_\omega \left(  1+
 K_1 \mathrm{e}^{K_2} \right)
\right\}. \qquad \label{eqn-c-hyp}
\end{eqnarray}
\end{myth}

\begin{proof}
Let $B_\lambda,B_\mu,B_\sigma,B_\omega,B_\theta,B_{\lambda'},B_{\mu'}> 0$  be arbitrarily large and $0<C_{\lambda}<B_{\lambda}$, $0<C_{\mu}<B_{\mu}$ be arbitrarily small, and consider the system (\ref{eq-plant2h1})--(\ref{eq-bc2h2}) when the control law uses the approximated gains:

\begin{equation}
U(t) = \int_0^1 \hat{k}_{u1}(\xi) u(\xi,t) d\xi + \int_0^1 \hat{k}_{v1}(\xi) v(\xi,t) d\xi, \label{eq-hyp2-approx-control}
\end{equation}
where \( \hat{k}_{u1}(\xi) = \hat{\mathcal{K}}_u(\lambda, \mu, \sigma, \omega, \theta, q)(\xi) \) and \( \hat{k}_{v1}(\xi) = \hat{\mathcal{K}}_v(\lambda, \mu, \sigma, \omega, \theta, q)(\xi) \).

Substituting this control law into the boundary condition of the target system \eqref{eq-hyp2-bc2}, we obtain
\begin{equation}
\beta(1,t) = \int_0^1 \tilde{k}_{u1}(\xi) u(\xi,t) d\xi + \int_0^1 \tilde{k}_{v1}(\xi) v(\xi,t) d\xi, \label{eq-hyp2-bc2-approx}
\end{equation}
where \( \tilde{k}_{u1}(\xi) = k_{u1}(\xi) - \hat{k}_{u1}(\xi) \) and \( \tilde{k}_{v1}(\xi) = k_{v1}(\xi) - \hat{k}_{v1}(\xi) \) are the approximation errors.

Using the inverse transformation \eqref{eq-hyp2-trans2}, we can express \( v(\xi,t) \) in terms of \( \beta(\xi,t) \) and \( u(\xi,t) \):
\begin{equation}
v(\xi,t) = \beta(\xi,t) + \int_0^\xi l_u(\xi,s) u(s,t) ds + \int_0^\xi l_v(\xi,s) \beta(s,t) ds.
\end{equation}

Substituting this expression into \eqref{eq-hyp2-bc2-approx}, we obtain
\begin{equation}
\beta(1,t) = \int_0^1 g_u(\xi) u(\xi,t) d\xi + \int_0^1 g_v(\xi) \beta(\xi,t) d\xi, \label{eq-hyp2-bc2-final}
\end{equation}
where
\begin{eqnarray}
g_u(\xi) &=& -\tilde{k}_{u1}(\xi) - \int_\xi^1 l_u(s,\xi) \tilde{k}_{v1}(s) ds, \label{eq-gu-def} \\
g_v(\xi) &=& -\tilde{k}_{v1}(\xi) - \int_\xi^1 l_v(s,\xi) \tilde{k}_{v1}(s) ds. \label{eq-gv-def}
\end{eqnarray}

 From Theorem~\ref{th-approx-hyp2}, given $B_\lambda,B_\mu,B_\sigma,B_\omega,B_\theta,B_{\lambda'},B_{\mu'},C_{\lambda},C_{\mu}$ positive, and for any $\epsilon>0$ such that $\epsilon<\epsilon^*$ with $\epsilon^*$ given in the Theorem~\ref{th-stab-hyp2} statement, there exists a neural operator that ensures that $\vert \tilde{k}_{u1}(x) \vert, \vert \tilde{k}_{v1}(x) \vert \leq \epsilon \quad \forall x\in[0,1]$. Thus we have the following bounds:
 
\begin{equation}
|g_u(\xi)| + |g_v(\xi)| \leq \epsilon \left(1 + \| l_u \|_\infty + \| l_v \|_\infty \right), \quad \forall \xi \in [0,1]. \label{eq-g-bound}
\end{equation}

Consider now the following Lyapunov function
\begin{eqnarray*}
V&=& \alpha \int_0^1 \frac{\mathrm{e}^{-cx}}{\lambda(x)} u^2 (x,t) dx
+\int_0^1 \frac{\mathrm{e}^{cx}}{\mu(x)} \beta^2 (x,t) dx
,
\end{eqnarray*}
for positive $\alpha$ to be chosen and $c$ as given by \eqref{eqn-c-hyp}. We get
\begin{eqnarray*}
 \dot V&\hspace{-3pt}=\hspace{-3pt}&
 -2\alpha \int_0^1 \mathrm{e}^{-cx} u (x,t) u_x(x,t)dx
% \nonumber \\ &&
 +
 2\alpha \int_0^1 \frac{\mathrm{e}^{-cx}}{\lambda(x)} u (x,t)
 \left(\sigma(x) u(x,t)+\omega(x) \beta(x,t) \right) 
 dx \nonumber \\ &&
 +2\alpha \int_0^1 \frac{\mathrm{e}^{-cx}}{\lambda(x)} u (x,t)
 \left(\int_0^x \kappa_u(x,\xi)u(\xi,t) d\xi
 %\right.\nonumber \\
%&&\left.
 +\int_0^x \kappa_v(x,\xi)\beta(\xi,t) d\xi \right) 
 dx \nonumber \\  &&
+2\int_0^1 \mathrm{e}^{cx} \beta (x,t) \beta_x(x,t) 
dx\nonumber \\  
&\hspace{-3pt}=\hspace{-3pt}&
-\alpha c \int_0^1 \mathrm{e}^{-cx} u^2(x,t)dx
-\alpha \mathrm{e}^{-c} u^2(1,t) 
%\nonumber \\ &&
+ \left(\alpha q^2-1 \right) \beta^2(0,t) 
+
 2\alpha \int_0^1 \frac{\mathrm{e}^{-cx}}{\lambda(x)} \sigma(x) u^2(x,t)
 dx 
 \nonumber \\ &&
 +2\alpha \int_0^1 \frac{\mathrm{e}^{-cx}}{\lambda(x)} 
 \omega(x) u (x,t)\beta(x,t) 
 dx
 %\nonumber \\ &&
 +2\alpha \int_0^1 \int_0^x  \frac{\mathrm{e}^{-cx}}{\lambda(x)} \kappa_u(x,\xi) u (x,t)
  u(\xi,t) d\xi
 dx
 \nonumber \\ &&
 +2\alpha \int_0^1\int_0^x \frac{\mathrm{e}^{-cx}}{\lambda(x)} \kappa_v(x,\xi) u (x,t)
  \beta(\xi,t) d\xi 
 dx
 %\nonumber \\  &&
-c \int_0^1 \mathrm{e}^{cx}  \beta^2 (x,t)  dx
\nonumber \\ &&
+ \mathrm{e}^{c}  \left( \int_0^1 g_u(\xi)u(\xi,t) d\xi+\int_0^1 g_v(\xi)\beta(\xi,t) d\xi\right)^2.
\quad
\end{eqnarray*}
 We get
\begin{eqnarray}
 \dot V&\leq &
-\alpha c C_{\lambda} t \int_0^1 \frac{\mathrm{e}^{-cx}}{\lambda(x)} u^2(x,t)dx + \left(\alpha q^2-1 \right) \beta^2(0,t) 
%\nonumber \\ &&
+
 2\alpha B_\sigma \int_0^1 \frac{\mathrm{e}^{-cx}}{\lambda(x)} u^2(x,t)
 dx 
 \nonumber \\ &&
 +2\alpha B_\omega \int_0^1 \frac{\mathrm{e}^{-cx}}{\lambda(x)} u (x,t)\beta(x,t) 
 dx
 %\nonumber \\ &&
 +2\alpha
 \Vert \kappa_u\Vert_{\infty}
 \int_0^1 \int_0^x  \frac{\mathrm{e}^{-cx}}{\lambda(x)}  u (x,t)
  u(\xi,t) d\xi
 dx
 \nonumber \\ &&
 +2\alpha 
 \Vert \kappa_v\Vert_{\infty} 
 \int_0^1 \int_0^x  \frac{\mathrm{e}^{-cx}}{\lambda(x)}  u (x,t)
  \beta(\xi,t) d\xi
 dx
% \nonumber \\  && 
-c C_{\mu} \int_0^1 \frac{\mathrm{e}^{cx}}{\mu(x)}  \beta^2 (x,t)  dx
\nonumber \\ &&
+ 2 \mathrm{e}^{c} \Vert g_u\Vert_\infty^2  \int_0^1 u^2(\xi,t) d\xi+
2 \mathrm{e}^{c} \Vert g_v\Vert_\infty^2 
\int_0^1 \beta^2(\xi,t) d\xi \nonumber  \\ &\leq & 
-\left(\alpha cC_{\lambda} -2\alpha B_\sigma- 2 \mathrm{e}^{2c} \Vert g_u\Vert_\infty^2  B_\lambda
-
\alpha B_\omega
%\nonumber \\ && \left.
-\alpha
 \Vert \kappa_u\Vert_\infty  \sqrt{\frac{B_\lambda}{C_{\lambda}}}-\alpha 
 \Vert \kappa_v\Vert_\infty 
\right)
\int_0^1 \mathrm{e}^{-cx} \frac{\mathrm{e}^{-cx}}{\lambda(x)}dx 
\nonumber \\ &&
-\left( cB_\mu-2 \mathrm{e}^{c} \Vert g_v\Vert_\infty^2  B_\mu
-\alpha \frac{B_\omega B_\mu}{C_\lambda} -\alpha 
 \Vert \kappa_v\Vert_\infty  \frac{B_\mu}{C_\lambda}
\right) 
%\nonumber \\ && \times 
\int_0^1 \frac{\mathrm{e}^{cx}}{\mu(x)}  \beta^2 (x,t)  dx
+ \left(\alpha q^2-1 \right) \beta^2(0,t)  \qquad
\end{eqnarray}
where we have used
\begin{eqnarray}
    && \int_0^1 \int_0^x  \frac{\mathrm{e}^{-cx}}{\lambda(x)} u (x,t)
  u(\xi,t) d\xi 
  \nonumber \\
  &\leq& \sqrt{\frac{B_{\lambda}}{C_{\lambda}}}
  \int_0^1 \int_0^x  \frac{\mathrm{e}^{-\frac{c}{2} x}}{\sqrt{\lambda(x)}} \vert u (x,t) \vert \frac{\mathrm{e}^{-\frac{c}{2} \xi}}{\sqrt{\lambda(\xi)}}
\vert  u(\xi,t) \vert d\xi  dx \nonumber \\
  &=&\sqrt{\frac{B_{\lambda}}{C_{\lambda}}}\frac{1}{2} 
  \int_0^1 \int_0^x  \frac{\mathrm{e}^{-\frac{c}{2} x}}{\sqrt{\lambda(x)}} \vert u (x,t) \vert \frac{\mathrm{e}^{-\frac{c}{2} \xi}}{\sqrt{\lambda(\xi)}}\vert  u(\xi,t) \vert d\xi  dx
  \nonumber \\   &&+\sqrt{\frac{B_{\lambda}}{C_{\lambda}}}\frac{1}{2} 
  \int_0^1 \int^1_\xi  \frac{\mathrm{e}^{-\frac{c}{2} x}}{\sqrt{\lambda(x)}} \vert u (x,t) \vert \frac{\mathrm{e}^{-\frac{c}{2} \xi}}{\sqrt{\lambda(\xi)}} 
\vert  u(\xi,t) \vert dx d\xi   \nonumber \\
    &=&\sqrt{\frac{B_{\lambda}}{C_{\lambda}}}\frac{1}{2}  \left( \int_0^1  \frac{\mathrm{e}^{-\frac{c}{2} x}}{\sqrt{\lambda(x)}} \vert u (x,t) \vert dx \right)^2  \nonumber \\
  &\leq&\sqrt{\frac{B_{\lambda}}{C_{\lambda}}} \frac{1}{2}  \int_0^1 \frac{\mathrm{e}^{-cx}}{\lambda(x)} u^2 (x,t)  dx \label{eqn-double}
\end{eqnarray}
where the second step in (\ref{eqn-double}) uses the fact that $\left(\int_0^1 f(x) dx\right)^2=\int_0^1 \int_0^1 f(x) f(\xi) dx d\xi=\int_0^1 \int_0^x f(x) f(\xi) dx d\xi+\int_0^1 \int_x^1 f(x) f(\xi) dx d\xi$, and
\begin{eqnarray}
  &&  \int_0^1 \int_0^x   \frac{\mathrm{e}^{-cx}}{\lambda(x)} u (x,t)
  \beta(\xi,t) d\xi 
  \nonumber \\
  &\leq& \frac{1}{\sqrt{C_\lambda}}
  \int_0^1 \int_0^x  \frac{\mathrm{e}^{-\frac{c}{2} x}}{\sqrt{\lambda(x)}} \vert u (x,t) \vert \mathrm{e}^{\frac{c}{2} \xi}
\vert  \beta(\xi,t) \vert \mathrm{e}^{-\frac{c}{2} (x+\xi)} d\xi  dx \nonumber \\
  &\leq&  \frac{1}{\sqrt{C_\lambda}}\left( \int_0^1 \frac{\mathrm{e}^{-\frac{c}{2} x}}{\sqrt{\lambda(x)}} \vert u (x,t) \vert dx \right)
   \left( \int_0^1  \mathrm{e}^{\frac{c}{2} x} \vert \beta (x,t) \vert dx \right) \nonumber \\
  &\leq& \frac{1}{2} \left(  \int_0^1  \frac{\mathrm{e}^{-c x}}{\lambda(x)}  u^2 (x,t)  dx+\frac{B_{\lambda}}{C_{\lambda}} \int_0^1  \frac{\mathrm{e}^{c x}}{\mu(x)}  \beta^2 (x,t)  dx   \right)
\end{eqnarray}
Choosing $\alpha=\min\{q^{-2},1\}\leq 1$ (so that $\alpha \geq \frac{1}{1+q^2}$),
\begin{eqnarray}
 \dot V&\leq & 
-\alpha \left( cC_{\lambda} -2 B_\sigma- \frac{2 \mathrm{e}^{2c} \Vert g_u\Vert_\infty^2  \bar \lambda}{1+q^2}
-
B_\omega -
 \Vert \kappa_u\Vert_\infty \sqrt{\frac{B_{\lambda}}{C_{\lambda}}}
% \right. \nonumber \\ && \left.
 -
 \Vert \kappa_v\Vert _\infty
\right) \int_0^1 \frac{\mathrm{e}^{-cx}}{\alpha(x)}  u^2 (x,t)  dx
 \nonumber \\ &&
-\left( cC_\mu-2 \mathrm{e}^{c} \Vert g_v\Vert_\infty^2  \bar \mu
%\right. \nonumber \\ && \left.
- \frac{B_\omega B_\mu}{C_\lambda} - 
 \Vert \kappa_v\Vert_\infty  \frac{ B_\mu}{C_\lambda}
\right) \int_0^1 \frac{\mathrm{e}^{cx}}{\mu(x)}  \beta^2 (x,t)  dx
\end{eqnarray}
Using the value of $\epsilon^*$ given in the statement of Theorem~\ref{th-stab-hyp2}, the bound of \eqref{eq-g-bound} and Theorem~\ref{th-kernelbound-hyp2} we get
\begin{eqnarray}
\Vert g_u \Vert_\infty^2 &\leq &\frac{cC_{\lambda}}{4B_\lambda}(1+q^2)\mathrm{e}^{-2c}\\
\Vert g_v\Vert_\infty^2 &\leq& \frac{cC_{\mu}}{4B_\mu}\mathrm{e}^{-c}
\end{eqnarray}
and thus we obtain
\begin{eqnarray}
 \dot V&\leq & 
-\alpha \left( \frac{cC_{\lambda}}{2} -2 B_\sigma
-
B_\omega -
 \Vert \kappa_u\Vert_\infty \sqrt{\frac{B_{\lambda}}{C_{\lambda}}}
% \right. \nonumber \\ && \left.
 -
 \Vert \kappa_v\Vert _\infty
\right) \int_0^1 \frac{\mathrm{e}^{-cx}}{\alpha(x)}  u^2 (x,t)  dx
 \nonumber \\ &&
-\left( \frac{cC_\mu}{2}
%\right. \nonumber \\ && \left.
- \frac{B_\omega B_\mu}{C_\lambda} - 
 \Vert \kappa_v\Vert_\infty  \frac{ B_\mu}{C_\lambda}
\right) \int_0^1 \frac{\mathrm{e}^{cx}}{\mu(x)}  \beta^2 (x,t)  dx
\end{eqnarray}
Since $ \Vert \kappa_u\Vert_\infty,  \Vert \kappa_v\Vert_\infty  \leq B_\omega K_1 \mathrm{e}^{K_2}$ (from their definition in \eqref{eq-kappa-def}), define now 
\begin{equation}
\delta=\max\left\{\delta_1,\delta_2 \right\}\,, \label{delta} 
\end{equation}
with $\delta_1=\frac{1}{C_\lambda} \left(2 B_\sigma
+
B_\omega \left(1+
 K_1 \mathrm{e}^{K_2} \left(1+ \sqrt{\frac{B_\lambda}{C_{\lambda}}} \right)\right)\right)$ and $\delta_2=\frac{B_\mu}{C_\lambda C_\mu}B_\omega \left(  1+
 K_1 \mathrm{e}^{K_2} \right)$
obtaining in the end
\begin{eqnarray}
 \dot V&\leq & 
-\alpha \frac{c-2\delta}{2} C_\lambda \int_0^1 \frac{\mathrm{e}^{-cx}}{\alpha(x)}  u^2 (x,t)  dx 
%\nonumber \\ &&
- \frac{c-2\delta}{2}C_\mu 
 \int_0^1 \frac{\mathrm{e}^{cx}}{\mu(x)}   \beta^2 (x,t)  dx
 \nonumber \\ &\leq& -\bar c V \label{lyap-hyp}
\end{eqnarray}
by using the definition of $c$ in \eqref{eqn-c-hyp}. Thus,  we get 
$L^2$ stability with user-fixed convergence rate $\bar c$ for the target system. The proof is finished by showing the same property for the original system by way of the direct/inverse exact transformation.
\end{proof}

\section{Conclusion}\label{sect-concl}

In this work, we explored the use of the DeepONet framework for computing gain kernels arising in the backstepping method for control of partial differential equations (PDEs). We introduced a novel methodology to directly approximate backstepping gains using neural operators and validated it across multiple case studies, including hyperbolic and parabolic plants. The efficacy of this method was critically examined in comparison to the previous approach that fully approximates the backstepping kernel, highlighting the advantages, disadvantages, and inherent challenges. 

A natural progression of this research is to extend the methodology to encompass other system configurations, such as coupled hyperbolic or parabolic systems, or higher-dimensional geometries like the $n$-dimensional ball. Although foundational challenges---such as discontinuous kernels in hyperbolic and parabolic designs, and the complexities of hyperspherical harmonics in the $n$-dimensional ball setting----need to be addressed, the groundwork established in this study suggests that these extensions can be methodically tackled. The potential adaptability to coupled designs is particularly intriguing. Managing the discontinuities in the gains that stem from piecewise continuous kernels will be crucial and may necessitate segmenting the kernels into multiple partitions for individual approximation. Another promising extension lies in the realm of developing observer gains. In this context, in-domain perturbations arise, as opposed to the boundary perturbations that appeared in this work. Nevertheless, we can expect similar complexities to those addressed herein, requiring the use of analogous Sobolev spaces and Lyapunov functionals, thereby potentially achieving similar results regarding observer convergence.

An intriguing avenue for future research lies in integrating neural operator approximations with the Fredholm backstepping approach to PDE control, pioneered by Jean-Michel Coron and collaborators~\cite{coron2014local,coron2015fredholm,coron2016stabilization,coron2018rapid}. The Fredholm method constructs integral transformations that map the original system to a target system with desired stability properties, involving Fredholm-type integral operators rather than the Volterra-type operators commonly used in traditional backstepping. The kernel functions in Fredholm backstepping are often more complex and challenging to compute due to the non-local nature of the integral operators and the mathematical intricacies involved in establishing kernel well-posedness. This presents both challenges and opportunities: for PDE mathematicians, there is the task of proving the continuity and Lipschitz properties of the operator mappings from plant parameters to kernels, especially when dealing with variable coefficients; for machine learning researchers, there is the potential to develop neural operators that approximate these complex kernel functions, addressing challenges such as generating training data by solving intricate kernel PDEs, devising effective training strategies, and evaluating the computational acceleration achieved over traditional numerical methods. By leveraging neural operator approximations in the Fredholm framework, we could significantly enhance computational efficiency, simplify the process of solving the associated kernel equations, and make real-time implementation more feasible, thereby expanding the scope of control strategies available for complex PDE systems.

\bibliography{NeuralGains_ArXiv}
\bibliographystyle{unsrt}
\end{document}